\documentclass{article}
\usepackage[squaren]{SIunits}
\usepackage{amsmath,amsthm,galois,amsfonts,amssymb,bm}
\usepackage{thm-restate}
\usepackage{todonotes}
\usepackage{algorithm,algpseudocode}
\usepackage{comment,paralist}
\usepackage[shortcuts]{extdash}

\usepackage[backend=bibtex]{biblatex}
\usepackage{dmbiblatex}
\bibliography{references}

\usepackage{listings}
\lstdefinelanguage{imp}[]{C}{morekeywords={assume,assert}}
\usepackage{url}
\newcommand{\mytitle}{A simple abstraction of arrays and maps by program translation}
\usepackage[pdfauthor={David Monniaux, Francesco Alberti},pdftitle={\mytitle}]{hyperref}
\newcommand{\functions}[2]{{#1} \rightarrow {#2}}
\newcommand{\parts}[1]{\mathcal{P}\left(#1\right)}
\newcommand{\abstr}[1]{{#1}^\natural}
\newcommand{\concr}[1]{{#1}^\flat}

\newcommand{\defn}{\stackrel{\vartriangle}{=}}
\newcommand{\ZZ}{\mathbb{Z}}
\newcommand{\NN}{\mathbb{N}}

\newcommand{\soft}[1]{\textsc{#1}}

\newcommand{\killspacebetweensectionandlisting}{\vspace{-1.5em}}

\newtheorem{lemma}{Lemma}

\begin{document}
\title{\mytitle}

\author{David Monniaux%
\thanks{The research leading to these results has received funding from the \href{http://erc.europa.eu/}{European Research Council} under the European Union's Seventh Framework Programme (FP/2007-2013) / ERC Grant Agreement nr.~306595 \href{http://stator.imag.fr/}{\mbox{``STATOR''}}}\\
{\small Univ. Grenoble Alpes, VERIMAG, F-38000 Grenoble, France}\\
{\small CNRS, VERIMAG, F-38000 Grenoble, France}
\and Francesco Alberti%
\thanks{This work has been carried out while the author was affiliated to the Universit\`a della Svizzera Italiana and supported by the Swiss National Science Foundation under grant no. P1TIP2\_152261.}\\
{\small Fondazione Centro San Raffaele, Milan, Italy}}

\maketitle

\begin{abstract}
We present an approach for the static analysis of programs handling arrays, with a Galois connection between the semantics of the array program and semantics of purely scalar operations.
The simplest way to implement it is by automatic, syntactic transformation of the array program into a scalar program followed analysis of the scalar program with any static analysis technique (abstract interpretation, acceleration, predicate abstraction,\dots).
The scalars invariants thus obtained are translated back onto the original program as universally quantified array invariants.
We illustrate our approach on a variety of examples, leading to the ``Dutch flag'' algorithm.

\end{abstract}

\section{Introduction}
\emph{Static analysis} aims at automatically discovering
program properties.
Traditionally, it has focused on dataflow properties (e.g. ``can this pointer be null?''), then on numerical properties (e.g. ``$2x+y \leq 45$ at every iteration of this loop'').
%
When it comes to programs operating over \emph{arrays}, special challenges arise.
For instance, the \soft{Astr\'ee} static analyzer,%
\footnote{\cite{DBLP:journals/fmsd/CousotCFMMR09,BlanchetCousotEtAl_PLDI03,BlanchetCousotEtAl02-NJ} \url{http://www.astree.ens.fr} \url{http://absint.de/astree/}}
based on abstract interpretation and commercially used in the avionics, automotive and other industries, supports arrays simplistically: it either ``smashes'' all cells in a single array into a single abstract value, or expands an array of $n$ cells into $n$ variables;
in many cases it is necessary to fully unroll loops operating over an array in order to prove the desired property%
\footnote{Possible since \soft{Astr\'ee} targets safety-critical embedded systems where array sizes are typically fixed at system design and dynamic memory allocation is prohibited.}.

In general, however, analyzing arrays programs entails exhibiting inductive loop invariants with universal quantification over array indices.
Neither smashing nor expansion
can prove, in general, that a simple initialization loop truly does work:
\begin{lstlisting}[caption={Simple array initialization},label={lst:init}]
int t[n]; for(int i=0; i<n; i++) t[i] = 0;
\end{lstlisting}
To derive the postcondition $\forall k. 0 \leq k < n \to t[k]=0$, one uses the loop invariant (in the Floyd-Hoare sense) $0 \leq i \leq n \wedge \forall k. 0 \leq k < i \to t[k]=0$.
The $0 \leq i \leq n$ part (or generalizations, e.g., filling the upper triangular part of a matrix) can be automatically inferred by many existing numeric analysis techniques.
In contrast, the $\forall k. 0 \leq k < i \to t[k]=0$ part is trickier and is the focus of this article.

\paragraph{Contribution}
We propose a generic method for analyzing array programs, which can be implemented
\begin{inparaenum}[i)]
\item
as a normal abstract domain
\item or by translating the program with arrays into a scalar program (a program without arrays), analyzing this program by any method producing invariants (back-end), and then recovering the array properties.
\end{inparaenum}
Its precision depends on the back-end analysis.
Our method has tunable precision and is formalized by Galois connections \cite{DBLP:journals/logcom/CousotC92} and,
contrary to most others, is not guided by a target property (here $\forall k. 0 \leq k < n \to t[k]=0$), though it can take advantage of it.
It can therefore be used to supply information to the end-user ``what does this program do?'' as opposed to be useful only for proving properties.
We demonstrate the flexibility of our approach on examples, using the acceleration procedure \soft{Flata}, the abstract interpreter \soft{ConcurInterproc} and \soft{CPAChecker} as back-ends.

We also show a form of \emph{completeness}: for any loop-free program, the precision of the analysis can be chosen so that it is exact with respect to universally quantified array properties (\S\ref{sec:completeness}). 

Our approach also applies to general maps $\mathit{keys} \rightarrow \mathit{values}$, though certain optimizations apply only to totally ordered index types.

\paragraph{Contents}
\hyperref[sec:galois]{Section~\ref*{sec:sentinel}} introduces our approach on one example.
\hyperref[sec:galois]{Section~\ref*{sec:galois}} discusses the Galois connections, and \autoref{sec:abstraction} gives the formal definition of our transformation algorithm and associated correctness and partial completeness proofs.
\hyperref[sec:examples]{Section~\ref*{sec:examples}} discusses the use of various backends on more examples. 
We finish with \hyperref[sec:related]{related work} and \hyperref[sec:conclusion]{conclusion}.
 

\section{Example: the Sentinel}
\label{sec:sentinel}
Our program transformation consists in
\begin{inparaenum}[i)]
\item a replacement of reads and writes parameterized by a number of distinguished indices, formalized in \autoref{sec:abstraction}
\item optionally, some ``focusing'' on a subset of index values
\item for certain backends (\textsc{ConcurInterproc}), the addition of observer variables implementing a form of partitioning.
\end{inparaenum}

\begin{lstlisting}[caption={A ``sentinel value'' marks the penultimate array cell},label=lst:sentinel]
  const int N=1000; int i = 0, t[N];
  initialize(N, t); t[N-2] = -1;
  while (t[i]>=0) i++;

\end{lstlisting}
Obviously to us humans, this program cannot crash with an array access out of bounds, and the final value of \lstinline|i| is, at most, 998 (its value depends on how the ``\lstinline|initialize|'' procedure works).
How can we obtain this result automatically?

Let \lstinline|x| be a symbolic constant in $\{ 0, \dots, N-1\}$.
We abstract array \lstinline|t| by the single cell \lstinline|t[x]|, represented by variable \lstinline|tx|: reads and writes at position \lstinline|x| in \lstinline|t| translates to reads and writes to variable \lstinline|tx| and reads and writes at other positions are ignored.
Program~\ref{lst:sentinel} is thus abstracted as:%
\footnote{We have left out, for the sake of brevity, tests for array accesses out of bounds.}
\lstinputlisting
{sentinel_transformed_short.txt}

\soft{Flata} \cite{BozgaIK10,DBLP:conf/fm/HojjatKGIKR12}
can compute an exact input/output relation of this program (to demonstrate generality, we left \lstinline|N| unfixed and replaced \lstinline|N-2| by a parameter \lstinline|p|; we thus use a precondition $0 \leq x < N \land 0 \leq p < N$):
{\scriptsize\begin{equation}\tag{$F$}
\begin{aligned}
   & (p=x \land i\leq x-1 \land i\geq 0 \land N\geq x+1) \lor
     (i=x \land i\geq 0 \land N\geq p+1 \land i\leq p-1) \lor\\
   & (x\geq p+1 \land i\leq x-1 \land i\geq 0 \land N\geq x+1 \land p\geq 0) \lor
     (i=x \land i\leq N-1 \land i\geq p+1 \land p\geq 0) \lor\\
   & (i\geq x+1 \land N\geq p+1 \land i\leq N-1 \land x\leq p-1 \land x\geq 0) \lor  \\
   & (i\leq x-1 \land i\geq 0 \land N\geq p+1 \land x\leq p-1) \lor\\
   & (i=x \land i=p \land i\geq 0 \land i\leq N-1) \lor
     (x\geq p+1 \land i\geq x+1 \land i\leq N-1 \land p\geq 0)\\[-1em]
\end{aligned}
\end{equation}}

Note that our abstraction is valid \emph{whatever the value of $x$}.
This means that $(i,p,N)$ should be a solution of
$N>0 \land \forall x~ (0\leq x<N \Rightarrow F)$.
One can check that this quantified formula entails $i \leq p$.

Arguably, we have done too much work: the only cell in the array whose content matters
much is at index $p$ (\lstinline|N-2| in the original program).
Running \textsc{Flata} with $x=p$ yields a postcondition implying $i \leq p$.
Again, this is sound, because \emph{any} choice of $x$ yields a valid postcondition on $(i,p)$.\medskip

\section{Galois connections}
\label{sec:galois}
We shall now see that, for any choice of indices, there is a Galois connection $\galois{\alpha}{\gamma}$ \cite{CousotC92} between the concrete (the set of possible values of the vector of variables of the original program) and the abstract set of states (the set of possible values of the vector of variables in the transformed program).
In general, this Galois connection is not onto: there are abstract elements $\abstr{x}$ that include ``spurious'' states, and which may be reduced to a strictly smaller $\alpha \circ \gamma(\abstr{x})$.

If $A$ and $B$ are sets, $\functions{A}{B}$ denotes the set of total functions from $A$ to $B$, and $\parts{A}$ the set of parts of~$A$. If $A$ is finite, $\functions{A}{B}$ denotes the set of \emph{arrays} indexed by $A$; specifically, if $A$ is $\{ 1, \dots, l_1 \} \times \dots \times \{ 1, \dots, l_d \}$ then $\functions{A}{B}$ denotes the $d$-dimensional arrays of size $(l_1,\dots,l_d)$.
$f[x]$ denotes the application $f(x)$ where $f$ is a program array or map.

Our constructions easily generalize to arbitrary combinations of numbers of arrays and numbers of indices; let us see a few common cases.

\subsection{Single index}
Applied with a single index, our map abstraction is classical \cite[\S2.1]{DBLP:conf/iccl/CousotC94}.

\paragraph{Definition}
Let $f \in \functions{A}{B}$, we abstract it by its graph
$\alpha_1(f) = \{ (a,f[a]) \mid a \in A \}$; e.g., a constant array $\{ 1, \dots, n \} \rightarrow \ZZ$ with value $42$ is abstracted as $\{(i,42) \mid 1 \leq i \leq n\}$.

We lift $\alpha_1$ (while keeping the same notation) to a function from $\parts{\functions{A}{B}}$ to $\parts{A \times B}$: for $\concr{F} \subseteq \functions{A}{B}$, $\alpha_1(\concr{F}) = \bigcup_{f \in F} \alpha_1(f)$, otherwise said
\begin{equation}
\alpha_1(\concr{F}) = \left\{ (a,f[a]) \mid a \in A, f \in \concr{F} \right\}
\end{equation}
Let $\abstr{F} \subseteq A \times B$. Then we define its concretization $\gamma_1(\abstr{F})$:
\begin{equation}
\gamma_1(\abstr{F}) = \left\{ f \in \functions{A}{B} \mid \forall a \in A~ (a,f[a]) \in \abstr{F} \right\}
\end{equation}
It is easy to see that $(\parts{\functions{A}{B}}, \subseteq) \galois{\alpha_1}{\gamma_1} \parts{A \times B}$ is a Galois connection.%

\paragraph{Non-surjectivity and reduction}
Remark that $\alpha_1$ is not onto (if $|A| > 1$ and $|B| > 0$): there exist multiple $\abstr{F}$ such that $\gamma_1(\abstr{F}) = \emptyset$, namely all those such that $\exists a \in A \forall b \in B~(a,b) \notin \abstr{F}$.
For instance, if considering arrays of two integer elements ($A = \{0,1\}$, $B = \ZZ$), then $\abstr{F} = \{ (1, 0) \}$ yields $\gamma_1(\abstr{F})=\emptyset$: there is no way to fill the array at index~$0$.

Let us now see the practical implication. Assume that the program has a single array in $\functions{A}{B}$ and a vector of scalar variables ranging in $S$, then the memory state is an element of $\concr{X} \defn S \times (\functions{A}{B})$. The scalar variables are combined into our abstraction as follows:
{\small%
\begin{equation}
\parts{S \times (\functions{A}{B})} \cong \functions{S}{\parts{\functions{A}{B}}} \galois{\alpha_1^S}{\gamma_1^S} \functions{S}{\parts{A \times B}} \cong \parts{S \times A \times B} \defn \abstr{X},
\end{equation}}%
where $\alpha_1^S$ and $\gamma_1^S$ lift $\alpha_1$ and $\gamma_1$ pointwise.
Let $s \in S$. While the absence of any $(s,a,b) \in \abstr{x}$ ($\abstr{x} \in 
\abstr{X}$) indicates that there is no $(s,f) \in \gamma^S_1(\abstr{x})$, that is, scalar state $s$ is unreachable, the converse is \emph{not} true. Consider a single integer scalar variable $s$ and an array $a$ of length~2, and $\abstr{x} = \{ (0,0,1), (1,0,0), (1,1,2) \}$, representing the triples $(s,i,a[i])$. It would seem that $s=0$ is reachable, but it is not, because there is no way to fill the array at position $1$: there is no element in $\abstr{x}$ of the form $(0,1,b)$.

A \emph{reduction} is a function $\rho: \abstr{X} \rightarrow \abstr{X}$ such that $\gamma \circ \rho = \gamma$ and $\rho(\abstr{x}) \subseteq \abstr{x}$ for all~$\abstr{x}$. The strongest reduction $\rho_{\text{opt}}$ (the minimum for the pointwise ordering induced by $\subseteq$) is $\alpha \circ \gamma$. In the above, $\rho_{\text{opt}}(\abstr{x}) = \{ (1,0,0), (1,1,2) \}$; intuitively, the strongest reduction discards all superfluous elements from the abstract value.

\paragraph{Class of formulas}
Assume now that the vector of scalar variables $s_1,\dots,s_m$ lies within $S=\ZZ^m$, the index $a$ lies in $\{ 1,\dots,l_1 \} \times \dots \times \{ 1,\dots,l_D \}$, and the values $f[a]$ also lie in~$\ZZ$.
Consider a formula $\psi$ of the form
\begin{equation}\label{eqn:form1}
\forall a_1,\dots,a_d~ \phi(s_1,\dots,s_m,a_1,\dots,a_d, f[a_1,\dots,a_d])
\end{equation}
where $\phi$ is a first-order arithmetic formula (say, Presburger).

Then, $f \models \psi$ if and only if
$\alpha^S_1(f) \subseteq \{ ((s_1,\allowbreak \dots,\allowbreak s_m),\allowbreak (a_1,\dots,\allowbreak a_d), \allowbreak b ) \allowbreak \mid \phi(s_1,,\allowbreak \dots,,\allowbreak s_m,a_1,,\allowbreak \dots, \allowbreak, a_d, b) \}$.
The sets of program states expressible by formulas of form~\ref{eqn:form1}
thus map through the Galois connection to a sub-lattice of $\parts{\ZZ^m \times \ZZ^d \times \ZZ}$.
This construction may be generalized to any theory or combination of theories over the sorts used for scalar variables, array indices, and array contents.

Checking that an invariant $\gamma^S_1(G)$ entails $\psi$, when the set $G$ is defined by a formula $\Gamma$, just amounts to checking that $\Gamma \land \neg\psi$ is unsatisfiable.

\subsection{Several indices, one per array}
The above settings can be extended to several arrays. 
Let $f,g \in \functions{A}{B}$, we abstract them by the product of their graphs
$\alpha_1(f,g) = \{ (a,f[a],a',g[a']) \mid a,a' \in A \}$,
$\gamma_1(\abstr{x}) = \{ (f,g) \in (\functions{A}{B})^2 \mid \forall a,a' \in A~ (a,f[a],a',g[a']) \in \abstr{x} \}$.
This abstraction can express properties of the form
{\small%
\begin{equation*}
\forall a_1,\dots,a_d,a'_1,\dots,a'_d~ \phi(s_1,\dots,s_m,a_1,\dots,a_d,
f[a_1,\dots,a_d],a'_1,\dots,a'_d,g[a'_1,\dots,a'_d])
\end{equation*}}%

As an example, the property that up to index $k$, monodimensional array $f$ of length $n$ has been copied into array $g$ can be expressed as
$\forall a,a' \in \{ 1,\dots,n \}~ a<k \land a=a' \Rightarrow f[a]=g[a']$ within that class.

\subsection{Dual indices, same array}
\label{sec:dual_indices}
\paragraph{Definition}
Let $f \in \functions{A}{B}$, pose
$\alpha_2(f) = \{ (a,f[a],a',f[a']) \mid a,a' \in A \}$
and lift it to a function from $\parts{\functions{A}{B}}$ to 
$\parts{(A \times B)^2}$.
Let $\abstr{F} \subseteq (A \times B)^2$. Then we define its concretization $\gamma_2(\abstr{F})$:
\begin{equation}
\gamma_2(\abstr{F}) = \left\{ f \in \functions{A}{B} \mid \forall a,a' \in A~ (a,f[a],a',f[a']) \in \abstr{F} \right\}
\end{equation}
It is easy to see that $(\parts{\functions{A}{B}}, \subseteq) \galois{\alpha_2}{\gamma_2} \parts{A \times B}$ is a Galois connection.

If $A$ is totally ordered, it seems a waste to include both $(a,f[a],a',f[a'])$ and $(a',f[a'],a,f[a])$ in the abstraction for $a<a'$. We thus define
$\alpha_{2<}(f) = \{ (a,f[a],a',f[a']) \mid a < a' \in A \}$ and
$\gamma_{2<}(\abstr{x}) = \{ f \in \functions{A}{B} \mid \forall a,a' \in A~, a<a
' \Rightarrow (a,f[a],a',f[a']) \in \abstr{x} \}$.

\paragraph{Non-surjectivity}
Remark, again, that $\alpha_2$ is not onto.
Consider an array of integers of length $3$, that is, a function $f: \functions{\{1,2,3\}}{\ZZ}$. An analysis computes its abstraction as $\abstr{x} = \{ (1, 0, 2, 0),\allowbreak (1, 0, 3, 0),\allowbreak (2, 0, 3, 0),\allowbreak (1, 0, 3, 1) \}$; recall that each element of that set purports to denote $(a,f[a],a',f[a'])$ for $a<a'$.
At first sight, it seems that $f(3)=1$ is possible, as witnessed by the last element. Yet, there is then no way to fill $a[2]$: there is no $x$ such that $(2,x,3,1) \in \abstr{x}$. This last element is therefore superfluous, and we can conclude that $\forall x~f[x]=0$. (See \S~\ref{sec:dutch_flag} for a real-life example.)

If $\abstr{x}$ is defined by a first-order formula ($\abstr{x} = \{ (a,b,a',b') \mid \phi(a,b,a',b') \}$), then this reduction (removing all $a',b'$ such that for some $a<a'$ there is no way to fill $f[a]$) is obtained as: $\forall a \exists b~ a<a' \Rightarrow \phi(a,b,a',b')$.

\paragraph{Class of formulas}
Assume now that the vector of scalar variables $s_1,\dots,s_m$ lies within $S=\ZZ^m$, the indices $a<a'$ lie in $\{ 1,\dots, n\}$, and the values $f[a],f[a']$ also lie in~$\ZZ$.
Consider a formula $\psi$ of the form $\forall a,a'~a<a' \Rightarrow \phi(s_1, \allowbreak \dots,\allowbreak s_m, \allowbreak a, \allowbreak f[a], \allowbreak a', \allowbreak f[a'])$ where $\phi$ is a first-order arithmetic formula (say, Presburger). For instance, one may express \emph{sortedness}: $\forall a,a'~ a<a' \Rightarrow f[a]\leq f[a']$.

Then, $f \models \psi$ if and only if
$\alpha^S_{2<}(f) \in \{ ((s_1,\dots,s_m),a,b,\allowbreak a',b' ) \mid \phi(s_1,\allowbreak\dots,\allowbreak s_m,\allowbreak a,b,\allowbreak a',b') \}$.
The sets of program states expressible by formulas of the form
$\forall a,a'~ a<a' \Rightarrow \phi(s_1,\dots,s_m,\allowbreak a,f[a],a',f[a'])$
thus map through the Galois connection to a sub-lattice of $\parts{\ZZ^m \times (\ZZ \times \ZZ)^2}$.


\section{Abstraction of program semantics}
\label{sec:abstraction}
Our analysis may be implemented by a syntactic transformation of array operations into purely scalar operations.
In this section, for each operation (read, write) we describe the transformed operation and demonstrate the correctness of the transformation.
We then discuss precision.

Without loss of generality, we consider only elementary reads and writes (\lstinline|r=f[i];| and \lstinline|f[i]=r;| with \lstinline|i| a variable). More complex constructs, e.g. \lstinline|f[e]=r;| with \lstinline|e| an expression, can always be decomposed into a sequence of scalar operations and elementary read and writes, using temporary variables.

\subsection{Transformation and Correctness}
\label{sec:correctness}
\paragraph{Reading from the array}
Consider a program state composed of $(s,r,i,f)$ where $r \in B$, $i \in A$ are scalars, $s \in S$ is the rest of the state, and $f \in \functions{A}{B}$. Consider the instruction \lstinline|r=f[i];|, its semantics is:
\begin{equation}
(s,r,i,f) \xrightarrow{\lstinline|r=f[i];|} (s,f(i),i,f)
\end{equation}
We wish to abstract it by the program fragment:
\begin{lstlisting}[caption={Read from array},label={lst:read_from_array1}]
r = random(); if (i==a) { r=b; }
\end{lstlisting}

\begin{restatable}{lemma}{alphagammareadi}
The forward and backward semantics of Program~\ref{lst:read_from_array1} abstract the forward and backwards semantics of \lstinline|r=f[i];| by the $(\alpha_1^S,\gamma_1^S)$ Galois connection.
\end{restatable}

More generally, a read with several indexes $a_1,a_2,\dots$ is abstracted by
\begin{lstlisting}[label={lst:read_from_array2}]
r=random();if (i==$a_1$) assume(r==$b_1$); if (i==$a_2$) assume(r==$b_2$); $\dots$
\end{lstlisting}
The same lemma and proof carry to that setting.

\paragraph{Writing to the array}
Consider the instruction \lstinline|f[i]=r;|, its semantics is:
\begin{equation}
(s,r,i,f) \xrightarrow{\lstinline|f[i]=r;|} (s,r,i,f[i \mapsto r])
\end{equation}
We wish to abstract it by the program fragment:
\begin{lstlisting}[caption={Write to array},label={lst:write_to_array1}]
if (i==a) { b=r; }
\end{lstlisting}

\begin{restatable}{lemma}{alphagammawritei}
The forward and backward semantics of Program~\ref{lst:write_to_array1} abstract the forward and backwards semantics of \lstinline|f[i]=r;| by the $(\alpha_1^S,\gamma_1^S)$ Galois connection.
\end{restatable}

\noindent The same carries over to writing to an array with several indices, abstracted as:
\begin{lstlisting}[caption={Write to array, multiple indexes},label={lst:write_to_array2}]
if (i==a1) { b1=r; }  if (i==a2) { b2=r; } $\dots$
\end{lstlisting}

\paragraph{Operations on scalars}
Consider a program state composed of $(s,f)$ where $f \in \functions{A}{B}$ is an array and $s \in S$ is the rest of the state. Consider a scalar instruction $s \xrightarrow{P} s'$ and thus $(s,f) \xrightarrow{\concr{P}} (s',f)$.
We abstract $P$ as: $(s,a,b) \xrightarrow{\abstr{P}} (s',a,b)$ if $s \rightarrow{P} s'$. Essentially, operations on scalars are abstracted by themselves.
The following result generalizes immediately to $(\alpha_2,\gamma_2)$ etc.

\begin{lemma}\label{ref:scalar_abstr_correct}
The forward and backward semantics of $\xrightarrow{\abstr{P}}$ abstract those of $\xrightarrow{\concr{P}}$ by the $(\alpha_1^S,\gamma_1^S)$ Galois connection.
\end{lemma}

\subsection{Precision loss}
``Forgetting'' the value of a scalar variable $v$ corresponds to $(s,v,f) \rightarrow (s,f)$.
This scalar operation may be correctly abstracted, as in \autoref{ref:scalar_abstr_correct}, by $(s,v,a,b) \rightarrow (s,a,b)$.
Surprisingly, applying this operation not only forgets the value of $v$, it may also enlarge the set of represented~$f$.

Example: $\abstr{x} = \{ (0,v,a,v) \mid a \in A \land v \in B \}$ abstracts by $(\alpha_1^S,\gamma_1^S)$ the set of triples $(0,v,f)$ where $f$ is a constant function of value $v$. Forgetting $v$ yields the set of pairs $(0,f)$ where $f$ is a constant function.
Applying $(s,v,a,b) \rightarrow (s,a,b)$ to $\abstr{x}$ yields
$\abstr{y} = \{ (0,a,v) \mid A \in A \land v \in B \}$, which concretizes
to the set $\{ (0,f) \mid f \in A \rightarrow B \}$. We have completely lost the ``constantness'' property.

\subsection{Relative completeness}
\label{sec:completeness}
We now consider the problem of \emph{completeness} of this abstraction, assuming that the back-end analysis is perfectly precise (thus \emph{relative completeness}).

Our analysis is incomplete in general.
Consider the following program:
\begin{lstlisting}[caption={Fill with zero, test zero},label={lst:zero_test_zero}]
  int t[N];  for(int i=0; i<N; i++) t[i]=0;
  for(int i=0; i<N; i++) if (t[i]!=0) break;

\end{lstlisting}
In the second loop, the \lstinline|break| statement is never reached and thus at the end of the loop, $i=N$.
Yet, if we distinguish $n < N$ different indices $i_1,\dots,i_n$, we cannot prove that this statement is never reached: for there will exist $i \in \{ 0,\dots,N-1 \} \setminus \{ i_1,\dots,i_n \}$ such that \lstinline|t[i]| returns, in the abstracted program, an arbitrary value and thus the \lstinline|break| statement is considered possibly reachable.

In contrast, when the program is loop-free, the abstraction is exact with respect to the scalar variables, provided the number of indices used for the abstraction is at least the number of array accesses:

\begin{restatable}{theorem}{loopfreecompleteness}
Consider a loop-free array program $P$ with arrays $a_1,\dots,a_d$ such that the number of accesses to these arrays are respectively $\alpha_1,\dots,\alpha_d$.
By abstracting these arrays with, respectively, $n_1,\dots,n_d$ indices such that $n_i \geq \alpha_i$ for all $i$, we obtain a Galois connection \galois{\alpha}{\gamma} such that $\pi_S \circ \gamma \circ \abstr{P} \circ \alpha = \pi_S \circ \concr{P}$ where $\pi_S$ is the projection of the state to the scalar variables.
\end{restatable}

This completeness results extends to universally quantified array properties $\forall i_1,\dots~ P(i_1, \dots) \rightarrow Q(a_1[i_1], \dots)$: one appends to the original program (assuming $i_1,\dots,i_n$ are fresh, nondeterministically initialized):
\begin{lstlisting}
assume($(P(i_1, \dots)$); assert($Q(i_1,\dots)$);
\end{lstlisting}
 

\section{More examples}
\label{sec:examples}

\subsection{Matrix initialization}
\begin{lstlisting}[caption={Initialization of $m\times n$ matrix $a$ with value $v$},label=lst:array_init_2d]
void array_init_2d(int m, int n, int a[m][n], int v) {
  for(int i = 0; i < m; i++) {
    for(int j = 0; j < n; j++) a[i][j] = v;            } }

\end{lstlisting}

Again, we consider cell $a[x,y]$, where $0 \leq x < m$ and $0 \leq y < n$, and disregard all other cells.
One should not convert this procedure into a single control-flow graph, because the resulting numerical transition system does not have the ``flat'' structure expected by \soft{Flata}~\cite{db2}.
Instead, one must encode the inner loop as a separate procedure:
\lstinputlisting
{array_init_2d_transformed2.txt}
\soft{Flata} then computes the exact input-output relation of \lstinline|inner_loop|, and finally the exact input-output relation of \lstinline|array_init_2d|:
{\small\begin{multline*}
(x=0 \land m=1 \land a'=v \land y \geq 0 \land n \geq y+1 )\lor
(a'=v \land x \geq 1 \land y \geq 0 \land m \geq x+1 \land n \geq y+1 )\lor\\
(n=1 \land x=0 \land y=0 \land a'=v \land m \geq 2 )\lor
(x=0 \land a'=v \land y \geq 0 \land m \geq 2 \land n \geq 2 \land n \geq y+1)
\end{multline*}}
Each disjunct implies $a'=v$, i.e., the final value of $a[x,y]$ is $v$. Again, because $(x,y)$ are symbolic constants with no assumption except that they are valid indices for $a$, this proves that all cells contain~$v$.
Assuming $0 \leq x < m \land 0 \leq y < n$ this formula may indeed be simplified automatically into $a'=v$.%
\footnote{We implemented a simplification algorithm for quantifier-free Presburger arithmetic inspired by \cite{Monniaux_LPAR08} so as to understand the output of \soft{Flata} and \soft{ConcurInterproc}.}
\medskip

\subsection{Slice initialization}\killspacebetweensectionandlisting

\begin{lstlisting}[caption={Initialize {$a[\mathit{low}\dots\mathit{high}-1]$} to {$v$}},label={lst:slice_init}]
void slice_init(int n, int a[n], int low, int high, int v) {
  for(int i=low; i<high; i++) a[i] = v;                    }
\end{lstlisting}

Again, we transform the program using a single index:
\lstinputlisting
{slice_init_transformed_short.txt}

\noindent \soft{Flata} produces as postcondition (assuming $0 \leq x <n \land 0 \leq \mathit{low} \leq \mathit{high} \leq n$):
{\small\begin{multline}
(\mathit{high}=\mathit{low} \land a'=a \land \mathit{high} \geq 0 \land n \geq \mathit{high} \land n\geq x+1 \land x \geq 0 )\lor\\
(a'=v \land \mathit{low} \leq x \land n \geq \mathit{high}\land \mathit{high}\geq x+1 \land \mathit{low} \geq 0 )\lor\\
(a'=a \land n \geq \mathit{high}\land \mathit{high}\geq low+1 \land \mathit{low} \geq x+1 \land x \geq 0 )\lor\\
(a'=a \land \mathit{high}\leq x \land n\geq x+1 \land \mathit{high}\geq low+1 \land \mathit{low} \geq 0)
\end{multline}}

Again, under the assumptions $0 \leq x < n$ and $0 \leq \mathit{low} \leq \mathit{high} \leq n$, this formula is equivalent to:
$
((\mathit{low} \leq x < \mathit{high}) \to a'=v) \land
(\neg (\mathit{low} \leq x < \mathit{high}) \to a'=a)
$.
Thus by quantification, the expected outcome:
\begin{multline}
(\forall x \in [\mathit{low}, \mathit{high})~ a'[x]=v) \land
(\forall x \notin [\mathit{low}, \mathit{high})
  \to a'[x]=a[x])
\end{multline}

\subsection{Array copy}\killspacebetweensectionandlisting

\begin{lstlisting}[caption={Copy array $a$ into array $b$},label=lst:array_copy]
void array_copy(int n, int a[n], int b[n]) {
  for(int i=0; i<n; i++) b[i] = a[i];      }
\end{lstlisting}

\noindent Take a single cell $a[x]$ in $a$ and a single cell $b[y]$ in $b$; after transformation:
\lstinputlisting
{array_copy_transformed.txt}

\paragraph{Flata}
\soft{Flata} yields:
$ (y \geq x+1 \land n \geq y+2 \land x \geq 0) \lor
(n=y+1 \land y \geq x+1 \land x \geq 0) \lor
(n=x+1 \land y \geq 0 \land y \leq x-1) \lor
(y \geq 0 \land y \leq x-1 \land n \geq x+2) \lor
(y=x \land b'=a \land n \geq x+2 \land x \geq 0) \lor
(y=x \land b'=a \land n=x+1 \land x \geq 0) $.
Assuming $0 \leq x < n \land 0 \leq y < n$, this is equivalent to $x=y \to a=b$. Thus by quantification, $\forall x, y . x=y \to a[x]=b[y]$, simplifiable into \bm{$\forall x . a[x]=b[x]$}.

\paragraph{Software model checking}
Many software model checkers, including \soft{CPAChecker}\footnote{\url{http://cpachecker.sosy-lab.org/}}, do not handle universally quantified array properties; yet we can use them as back-end analyses!
We translate the target property (here $\forall x. 0\leq x<n \to a[x]=b[x]$) into a precondition $x=y$ and an assertion on the postcondition $a=b$.
\soft{CPAChecker} then proves the property.%
\footnote{\texttt{scripts/cpa.sh -predicateAnalysis} after preprocessing with \texttt{assert.h}}

\lstinputlisting
{array_copy_cpachecker.txt}

\subsection{In-place array reversal}\killspacebetweensectionandlisting
\label{sec:reversal}

\begin{lstlisting}[caption={Array reversal},label=lst:array_reverse_inplace]
void array_reverse_inplace(int n, contents t[n]) {
  int i=0, j=n-1;
  while(i < j) {
    contents tmp1 = t[i], tmp2 = t[j];
    t[i] = tmp2; t[j] = tmp1;     i++; j--;    } }

\end{lstlisting}

For this program, we need to distinguish the initial values in the array from the values during the computation (which finally yield the final values).
We use three indices $0 \leq x < n$, $0 \leq y \leq z < n$: $a$ is the initial value of $t[x]$, $b$ the current value of $t[y]$, $c$ the current value of~$t[z]$.

For each read, we check if the index of the read is equal to $y$ (respectively, $z$) and return $b$ (respectively, $c$) if this is the case.
If the index is equal to both $y$ and $z$, it is sound to return either $b$ or $c$; we chose to return $b$.
For each write, we test if the index is equal to $y$, in which case we write to $b$, and equal to $z$, in which case we write to $c$.
If it is equal to both $y$ and $z$, we write to both $b$ and $c$.

\begin{lstlisting}[caption={Array reversal, transformed},label=lst:array_reverse_inplace_transformed]
  contents a, b, c;
  int x, y, z, i=0, j=n-1;
  if (y == x) b = a;  if (z == x) c = a;
  while(i < j) { contents tmp1, tmp2;
    if (i == y) tmp1 = b;  else if (i == z) tmp1 = c;
    if (j == y) tmp2 = b;  else if (j == z) tmp2 = c;
    if (i == y) b = tmp2;  if (i == z) c = tmp2;
    if (j == y) b = tmp1;  if (j == z) c = tmp1;     i++; j--; }
\end{lstlisting}

\paragraph{Flata}
\soft{Flata} takes {480\,\second}%
\footnote{All timings using one core of a {2.4\,\giga\hertz} {Intel \circledR} {Core\texttrademark} i3 running 32-bit Linux.}
to process this program, and outputs an input-output relation $\phi$ in disjunctive normal form with 292 disjuncts (not reprinted).
The output formula is very complicated, with explicit enumeration of many particular cases; the reason for the slowness and the size of the output formula seems to be that \soft{Flata} explicitly enumerates many cases up to saturation, with no attempt at intermediate simplifications.
We shall now explain what this formula entails.

Let $U$ be $0 \leq x,y,z < n \land y+z=n-1$.
Let $U_<$ be $U \land y < z \land z = x \land y+z=n-1$, then $\phi \land U_<$ is equivalent to $a=b \land U_<$.
This means that under the precondition $U_<$, Prog.~\ref{lst:array_reverse_inplace_transformed} has exact postcondition $a=b$.
By universal quantification, this means that $\forall x,y,z . U_< \to t[x]=t'[y]$, where $t$ is the input array to Prog.~\ref{lst:array_reverse_inplace} and $t'$ the output.
This formula may be simplified into
$
\forall x . 0 \leq x \land 2x \leq n-2 \to t[x] = t'[n-1-x]
$;
We can obtain similar formulas for the cases $y > z$ and $y = z$.
The three cases can can be summarized into
\begin{equation}\label{eqn:reverse}
\bm{\forall x . 0 \leq x < n \to t[x] = t'[n-1-x]}
\end{equation}

\paragraph{Flata, focused}
The above execution time and the complexity of the resulting formula seem excessive, if all that matters is when $(x=y \lor x=z) \land y+z=n-1$.
Indeed, some easy static analysis (by \soft{Flata} or another tool) shows that the array accesses within the loop are done at indices $i$ and $j$ that satisfy $0 \leq i \leq j < n$ and $i+j=n-1$.
Such a pre-analysis suggests to target the main analysis to two positions $t[y]$ and $t[z]$ in the current array, satisfying $0 \leq y \leq z < n$ and $y+z=n-1$.
The only positions $a[x]$ that matter in the original array are those that can be read precisely, that is, $x=y$ and $x=z$.

We therefore re-run the analysis with precondition $U$: $(0 \leq y \leq z < n \land y+z=n-1 \land x=y)$. \soft{Flata} runs for {6\,\second} and outputs a formula with 8 disjuncts, with $a=c$ in all disjuncts.
We thus have proved that
$\forall x,y,z. U \to t[x]=t'[z]$, which can be simplified into
\bm{$\forall z . 2z \geq n-1 \land z < n \to t'[z]=t[n-1-z] $}.

We may also run with the precondition, $(0 \leq y \leq z < n \land y+z=n-1 \land x=z)$ and get the remainder of the cases to conclude as in Formula~\ref{eqn:reverse}.

To summarize, when the exact analysis of the transformed program (that is, an exact analysis in the back-end) is too costly, one may choose to \emph{focus} the analysis by restricting the range of the indices ($x,y,z,\dots$) to some area $U$ considered to be ``meaningful'', for instance obtained by pre-analysis of the relationships between the indices of the array accesses in the program.
This is sound, since the quantification in the resulting formula is over the indices satisfying $U$.
Thus, a bad choice for $U$ may only result in a sound, but uninteresting invariant (the worst case is to take an unsatisfiable $U$: we then obtain a formula talking about an empty set of positions in the arrays, thus a tautology).

\paragraph{ConcurInterproc, focused}
\soft{Interproc}%
\footnote{\url{http://pop-art.inrialpes.fr/people/bjeannet/bjeannet-forge/interproc/}}
applies classical abstract interpretation (Kleene iteration accelerated with widenings, with possible narrowing iterations) over a variety of numerical abstract domains provided by the \soft{Apron} \cite{DBLP:conf/cav/JeannetM09} library\footnote{\url{http://apron.cri.ensmp.fr/library/}}
(intervals, ``octagons'' \cite{DBLP:journals/lisp/Mine06}, convex polyhedra \cite{halbwachs:tel-00288805,CousotH78}\dots).

\soft{ConcurInterproc}%
\footnote{\url{http://pop-art.inrialpes.fr/interproc/concurinterprocweb.cgi}}
extends it to concurrency (which we will not use here) and partitioning of the state space according to enumerated types, including Booleans.
In a nutshell, while \soft{Interproc} assigns a single abstract element (product of intervals, octagon, polyhedron) to each program location, \soft{ConcurInterproc} attaches $2^n$ abstract elements, where $n$ is the number of Booleans (or,more generally, one per concrete instantiation of the enumerated variables).
In order to achieve this at reasonable cost, the \soft{BDDApron} library uses a compact representation, where identical abstract elements are shared and the associated set of concrete instantiations is represented by a binary decision diagram.

Program~\ref{lst:array_reverse_inplace_transformed} contains no Boolean variable (or of any other enumerated type), thus directly applying \soft{ConcurInterproc} over it will yield one convex polyhedron at the end;
yet we need to express a disjunction of such polyhedra (e.g. there is the case where $x=y$, and the case where $x \neq y$, which may be subdivided into $x < y$ and $y < z$).
Furthermore, inside the loop one would have to distinguish $i < y$, $i = y$, $i > y$.
This is where, in other analysis of array properties by abstract interpretation  \cite{GopanRS05,HalbwachsP08,peron:tel-00623697,perrelle:tel-00973892,CousotCL11} one introduces ``slices'' or ``segments'' of programs, often according to syntactic criteria.
In our case, we wish to distinguish certain locations in the array (or combinations of several locations, as here with three indices $x,y,z$) according to more semantic criteria.

Our solution is to introduce \emph{observer} variables, which are written to but never read and whose final value is discarded, but which will guide the analysis and the partitioning performed.
Here, we choose to have one flag variable per access, initially set to ``false'', and set to ``true'' when the access has taken place.
As previously, we use a precondition $y+z=n-1 \land x=z$.

{
\begin{lstlisting}[basicstyle={\fontfamily{ppl}\selectfont\scriptsize},caption={Array reversal, transformed and instrumented},label=lst:array_reverse_inplace_instrumented]
  contents a, b, c;
  int x, y, z;
  bool y0,z0,y1,z1,y2,z2,y3,z3,y4,z4;
  x0=y0=y1=z1=y2=z2=y3=z3=y4=z4=false;
  int i=0, j=n-1;
  assume(y+z == n-1); assume(x==z);
  if (y == x) { b = a; y0 = true; } if (z == x) { c = a; z0 = true; }
  while(i < j) {
    contents tmp1, tmp2;
    if (i == y) {tmp1 = b; y1 = true;} else if (i == z) {tmp1 = c; z1 = true;}
    if (j == y) {tmp2 = b; y2 = true;} else if (j == z) {tmp2 = c; z2 = true;}
    if (i == y) {b = tmp2; y3 = true;} if (i == z) {c = tmp2; z3 = true;}
    if (j == y) {b = tmp1; y4 = true;} if (j == z) {c = tmp1; z4 = true;}
    i++; j--;                                                                  }
\end{lstlisting}
}

\noindent\soft{ConcurInterproc}, within 0.16\,\second, concludes that $a=b$.

\subsection{Dutch national flag}
\label{sec:dutch_flag}
\emph{Quicksort} is a divide-and-conquer sorting algorithm: pick a \emph{pivot}, swap array cells until the array is divided into two areas: elements less than the pivot, and elements greater than or equal to it; then recurse in both areas.
An improvement, in case many elements may be identical, is to swap the array into three areas: elements less than the pivot, equal to it, and greater than it, and recurse in the ``less'' and ``greater'' areas.
This \empty{three-way partition} is equivalent to the ``Dutch national flag problem'' \cite[ch.~14]{Dijkstra76}, of swapping pebbles of colors red, white and blue (corresponding to ``less'', ``equal'' and ``greater'') into three segments.

\begin{lstlisting}[label={lst:dutch_flag},caption={Dutch flag\protect\footnotemark}]
void threeWayPartition(int data[], int size, int low, int high){
 int p = -1, q = size;
 for (int i = 0; i < q;) {
  if(data[i] < low) {swap(&data[i], &data[++p]); ++i;}
  else if(data[i]>=high) {swap(&data[i], &data[--q]);} else ++i;
}}

\end{lstlisting}
\footnotetext{Courtesy of Wikipedia}

We transform this program with two indices $0 \leq x < y < n$ (remark that this is valid only if $n \geq 2$) with associated values $\mathit{datax}$ and $\mathit{datay}$, and instrument it with Boolean observer variables: for each read or write access to an index $i$, we keep a Boolean recording the value of predicate $x \leq i$ and one for $x \geq i$ (respectively for~$y$).
The values in the array are encoded as pebble colors LOW, MIDDLE, HIGH.


\soft{ConcurInterproc} computes a postcondition within 1\,\minute. The resulting formula $\phi$ has 52 cases; we will not print it here.
We check that $\phi \land x \leq p \to \mathit{datax} = \text{BLUE}$, meaning that finally,
$\bm{\forall x. 0 \leq x \leq p \to t[x] = \text{\textbf{BLUE}}}$
Similarly, $\phi \land y \geq q \to \mathit{datay} = \text{RED}$, thus
$\bm{\forall y. q \leq y < n \to t[y] = \text{\textbf{RED}}}$.
We would expect as well that $\forall x. p < x < q \to t[x] = \text{WHITE}$.
Yet, this does not immediately follow from~$\phi$:
$\phi \land p < y < q \land \mathit{datay} = \text{RED}$ is satisfiable!
Could there be red cells in the supposedly white area?

Note that $\phi$, for fixed values of $n,p,q$, encodes quadruples $(x,\mathit{datax},y,\mathit{datay})$, which encompass all possible values of $(x,t[x],y,t[y])$ for $x < y$.
In particular, for $t[y]=\text{RED}$ to be possible for given $n,p,q$, one must have suitable $t[x]$ for all $x < y$, such that $(x,t[x],y,\text{RED})$ satisfies $\phi$ for the same $n,p,q$.
In other words, to have a cell $t[y]=\text{RED}$ one must be able to find values $t[x]$ for all cells to the left of it.
We check that, indeed, $p < y < q \land \mathit{datay} \neq \text{WHITE} \land (\forall x. 0 \leq x < y \to \phi)$ is unsatisfiable,%
\footnote{From Presburger arithmetic, a decidable theory.}
meaning that
$\forall y. (p < y < q \land y > 0) \to t[y] = \text{WHITE}$.
Furthermore, $\phi \land x=0 \land x<q \land \mathit{datax} \neq \text{WHITE}$ has no solution.
We can thus conclude
$\bm{\forall y. p < y < q \to t[y] = \text{\textbf{WHITE}}}$.

Thus, we encountered a case of ``spurious'' solutions in the abstract element, due to the fact that the abstraction is not onto and that certain abstract elements can be reduced to a smaller element with the same concretization; which was achieved through quantification (see~\autoref{sec:dual_indices}).
This reduction can thus be performed through some form of \emph{quantifier elimination}.

\section{Related work}
\label{sec:related}

\paragraph{Acceleration}
For certain classes of loops, it is possible to compute exactly the transitive closure $\tau^+$ of the relation $\tau$ encoding the semantics of the loop, within a decidable class.
Acceleration for arrays has been studied by Bozga et al. \cite{BozgaHIKV09},
who obtain the transitive closure in the form of a \emph{counter automaton}.
The translation from counter automaton to array properties expressed in first-order logic then requires an abstraction step, resulting in a loss of precision.
Alberti et al. \cite{AlbertiGS13,AlbertiGS14} proposed a template-based solution.
Certain classes of $\tau$'s admit a definable
acceleration in Presburger arithmetic augmented with free function
symbols, at the price of nested quantifiers.
The $\exists^*\forall^*$ fragment of this theory is undecidable \cite{Halpern91}; thus again abstraction is needed to apply this technique in practice.
Yet, there are cases where exact acceleration is possible~\cite{AlbertiGS15}.
Contrary to these approaches,
\begin{inparaenum}[i)]
\item ours does not put restrictions on the shape of the loop (and the program in general)
\item we perform the tunable abstraction first, with the rest of the analysis being delegated to a back-end (which can possibly use exact acceleration on scalar
programs~\cite{BozgaIK10}).
\end{inparaenum}

\paragraph{Abstract interpretation}
Various array abstractions  \cite{GopanRS05,HalbwachsP08,peron:tel-00623697,perrelle:tel-00973892,CousotCL11} distinguish \emph{slices}
or \emph{segments}, whose contents is then abstracted by another
abstract domain.
Depending on the approach, relationships between several slices may or
may not be expressed, and the partitioning may be syntactic or based on
some pre-analysis.
To our best knowledge,
none of these approaches work on multidimensional arrays or on
maps, contrary to ours.
One major difference between these approaches and ours is that ours
separates the analysis, both in theory and implementation, into
an abstraction that maps array programs to scalar programs and an
analysis for the scalar programs, while theirs are more ``monolithic''.
Even though they are parametric in abstract domains for values and
possibly indexes, they must be used inside an abstract interpreter based on
Kleene iterations with widening.
In contrast, ours can use any back-end analysis for scalar programs,
including exact acceleration, abstract interpretation
with Kleene iterations, policy iteration, and even, if a target property
is supplied, predicate abstraction (see CEGAR below).

Cox et al. \cite{DBLP:conf/sas/CoxCR14} do not target array programs per se, but
programs in highly dynamic object-oriented languages such as Javascript,
where an object is a map from fields to values and the
set of possible field names is not fixed.
Dillig et al. \cite{DilligDA10} overcome the dichotomy of strong
vs weak updates with \emph{liquid updates}.
Their approach is monolithic and cannot express properties such as sortedness.

\paragraph{Predicate abstraction and CEGAR}
\emph{Predicate abstraction} starts from the control structure of a
program and incrementally refines it by splitting control states
according to predicates chosen by the user \cite{FlanaganQ02} or, commonly, obtained by counterexample\-/guided abstraction refinement (CEGAR).
From an abstract counterexample trace not corresponding to a concrete counterexample,
they refine the model using local predicates constituting a step-by-step proof that 
this abstract trace does not match any concrete trace.
The hope is that this proof generalizes to more counterexample traces
and that the predicates eventually converge to define an inductive
invariant.
The predicates are obtained from \emph{Craig interpolants}
\cite{DBLP:conf/apn/McMillan05,McMillan06,McMillan11}
extracted from the proof of
unsatisfiability produced by a
\emph{satisfiability modulo theory} (SMT) solver.
The difficulty here is to generate Craig interpolants that tend to
generalize to inductive invariants, on \emph{quantified} formulas
involving arrays \cite{DBLP:conf/tacas/McMillan08}. We are interested in
predicates such as $\forall 0 \leq k < i,~t[k]=0$, which generalizes to
an inductive invariant on Program~\ref{lst:init}, as opposed to, say,
$t[0]=0 \land t[1]=0$, which is equivalent for $i=2$ but does not
generalize to arbitrary~$i$.
In order to achieve practical scalability, some work restrict
themselves to the inference of array predicates to certain forms, e.g.
\emph{range predicates} \cite{JhalaM07}.
Others tune the interpolating
procedure towards the generation of better interpolants
\cite{AlbertiBGRS14,Alberti_Monniaux_SAC-SVT2015}.
A major difference between our approach and those based on CEGAR is that
we do not require a ``target'' property to prove, which is necessary for
having counterexamples, though we can use one if needed.
If such a property is provided, our approach can use as a back-end a CEGAR system limited to scalar variables.

\paragraph{Theorem proving and SMT-based approaches}
The generation of invariants for programs with arrays has been also
studied using automated theorem proving \cite{HoderKV11,HoderKV10};
this approach is generally limited by the fact that theory reasoning (e.g. arithmetic) and superposition-based deductive reasoning (on which the Vampire first-order theorem prover is based \cite{KovacsV13}) are not yet efficiently integrated.
As opposed to \cite{BjornerMR13}, we do not rely on
quantifier-instantiation procedures.



\paragraph{Quantification}
Flanagan et al. \cite{FlanaganQ02} also use Skolem constants that they quantify
universally after analysis steps.
As opposed to us, they require the user to specify the
predicates on which the program will be abstracted.

\paragraph{Abstraction of sets of maps}
Our approach generalizes a classical abstraction of sets of maps
\cite[\S2.1]{DBLP:conf/iccl/CousotC94}.
Jeannet et al. \cite{DBLP:conf/sas/JeannetGR05} considered the problem of abstracting
sets of functions of signature $D_1 \rightarrow D_2$, assuming a
\emph{finite} abstract domain $A_1$ of cardinality $n$ abstracting
subsets of $D_1$ and an abstract domain $A_2$ abstracting subsets
of~$D_2^n$. In contrast, we do not make any cardinality assumption.

\paragraph{Partitioning}
Rival et al. \cite{DBLP:journals/toplas/RivalM07} introduced partitioning according to an abstraction of the history of the computation.
Our approach using observer variables for using \soft{ConcurInterproc} (\autoref{sec:reversal}) is akin to considering a finite abstraction of the trace of read/writes into a given array.


\section{Conclusion and Future Work}
\label{sec:conclusion}
We have shown that a number of properties of array programs can be proved by abstracting the array $a$ using a few symbolic cells $a[x], a[y], \dots$ by automatically translating the program into a scalar program, running a static analyzer over the scalar program and translating back the invariant for the original program.
In some cases, a form of quantifier elimination is used over the resulting formulas.

Our approach is not specific to arrays, and can be applied to any map structure $X \rightarrow Y$ (e.g. hash tables and other container classes).
A possible future extension is multiset properties, a multiset being map $X \rightarrow \NN$.

The main weakness of our approach is the need for a rather precise back-end analysis (for the scalar program obtained by translation).
Our experiments highlighted some inefficiencies in e.g. \soft{Flata} and \soft{ConcurInterproc}:
in the former, many paths can be enumerated and complicated formulas generated even though a much simpler equivalent form exists;
in the latter, polyhedra that are only slightly different (say, one constraint is different) are handled wholly separately.
This gives immediate directions for research for improving exact acceleration, as in \soft{Flata}, or disjunctions of polyhedra, as in \soft{ConcurInterproc}.
Another difficulty, if using \soft{ConcurInterproc} or other tools focusing on convex sets of integer vectors, is the need to use observer variables and/or an auxiliary pre-analysis to ``focus'' the main analysis.

We stress again that we obtained our results using unmodified versions of very different back-end analyzers (\soft{ConcurInterproc}, \soft{Flata}, \soft{CPAChecker}), which testifies to the flexibility of our approach.
Performance and precision improvements can be expected by modifying the back-end analyzers (e.g. precision could be improved by performing reduction steps during the analysis, rather than after the computation of the invariants).


\printbibliography

\appendix
\section{Proofs}
\alphagammareadi*
\begin{proof}
Consider an abstraction $\abstr{x} \subseteq S \times B \times A \times (A \times B)$ of $(s,r,i,f)$: $\forall a\in A~ (s,r,i,a,f[a]) \in \abstr{x}$.
The image of the set $\abstr{x}$ by that program is
$\abstr{y} = \{ (s, r', i, a, b) \mid r' \in B \land i \neq a \land (s, r, i, a, b) \in \abstr{x} \} \cup \{ (s, b, i, i, b) \mid (s, r, i, i, b) \in \abstr{x} \}$.
It is clear that $(s,f(i),i,f) \in \gamma(\abstr{y})$, otherwise said $\forall a\in A~ (s,f(i),i,a,f[a]) \in \abstr{y}$.

The pre-image of the set $\abstr{x}$ by that program is
$\abstr{z} = \{ (s, r, i, a, b) \mid r \in B \land i \neq a \land (s,r',i,a,b) \in \abstr{x} \} \cup \{ (s, r, i, i, b) \mid r \in B \land (s, b, i, i, b) \in \abstr{x} \}$.
Assume $(s,r',i,f) \in \gamma(\abstr{x})$ and $(s,r,i,f) \xrightarrow{\lstinline|r:=f[i]|} (s,r',i,f)$; then
\begin{inparaitem}
\item either $r' \neq f(i)$: then there is no such $(s,r,i,f)$, thus any such $(s,r,i,f) \in \gamma(\abstr{z})$Th
\item either $r' = f(i)$, then any $(s,r,i,f)$ fits; let us now prove $(s,r,i,f) \in \gamma(\abstr{z})$: let $a \in A$, then either $i=a$ and $(s,r,i,a,f[a]) \in \abstr{z}$ (second disjunct), or $i \neq a$ and $(s,r,i,a,\allowbreak f[a]) \in \abstr{z}$ (first disjunct).
\end{inparaitem}
\end{proof}

\alphagammawritei*
\begin{proof}
Consider an abstraction $\abstr{x} \subseteq S \times B \times A \times (A \times B)$ of $(s,r,i,f)$: $\forall a\in A~ (s,r,i,a,f[a]) \in \abstr{x}$.
The image of the set $\abstr{x}$ by that program is
$\abstr{y} = \{ (s,r,i,a,b) \mid i\neq a \land (s,r,i,a,b) \in \abstr{x} \}
\cup \{ (s,r,i,i,r) \mid  (s,r,i,a,b) \in \abstr{x} \}$.
Let us prove that $(s,r,i,f[i \mapsto r]) \in \gamma(\abstr{y})$.
Let $a \in A$. If $a \neq i$, then $\left(s,r,i,a,f[i \mapsto r](a)\right)=(s,r,i,a,f(a)) \in \abstr{y}$ (first disjunct); if $a=i$, then $\left(s,r,i,a,f[i \mapsto r](a)\right)=(s,r,i,i,r) \in \abstr{y}$ (second disjunct).

The pre-image of the set $\abstr{x}$ by that program is
$\abstr{z} = \{ (s,r,i,i,b' \mid b' \in B \land (s,r,i,i,b) \in \abstr{x} \}\cup
             \{ (s,r,i,a,b) \mid i \neq a \land (s,r,i,a,b) \in \abstr{x} \}$.
Assume $(s,r,i,f') \in \gamma(\abstr{x})$ and $(s,r,i,f) \allowbreak \xrightarrow{\lstinline|r:=f[i]|} \allowbreak (s,r,i,f')$; let us prove $(s,r,i,f) \in \gamma(\abstr{z})$.
Let $a \in A$. If $a=i$, then $(s,r,i,i,f(i)) \in \abstr{z}$ (first disjunct)
If $a\neq i$, then $(s,r,i,a,f(a))=(s,r,i,a,f'(a)) \in \abstr{z}$ (second disjunct).
\end{proof}

\loopfreecompleteness*
\begin{proof}
Consider an execution trace $T$ in $P$, and record the indices $\xi_{i,j}$ of the $j$-th (numbered syntactically) access to the $i$-th array.
Consider now the program $P'$ obtained by abstracting $P$ according to $\alpha_i$ indices for each array $a_i$, i.e. each read $r:=a_i[e]$ is transformed into
\begin{lstlisting}
r = random();
if ($e$==$x_{i,1}$) { assume(r==$b_{i,1}$); }  if ($e$==$x_{i,2}$) { assume(r==$b_{i,2}$); }
...
\end{lstlisting}
and each write $a_i[e]:=w$ as
\begin{lstlisting}
if ($e$==$x_{i,1}$) { $b_{i,1}$ = $w$; }  if ($e$==$x_{i,2}$) { $b_{i,2}$ = $w$; } $\dots$
\end{lstlisting}

Now replay $T$ in $P'$, with the same initial values, the same external and nondeterministic choices, and $x_{i,j}=\xi_{i,j}$.
Then, for any array access in the execution of $P'$, at least one of the tests is taken (the program does not fall into the case where none of the selected indices match the index for the read/write instruction).
In the case of a read $r:=a_i[e]$, the value read in $P'$ is then the same as the one read in $P$. Then, the execution of $P'$ faithfully mimics that of $P$.
The final values for the execution of $T$ in $P'$ are thus the same as those in $P$, which proves the statement.
\end{proof}

\end{document}